\documentclass[11pt]{article}

\usepackage[margin=1in]{geometry}
\usepackage{microtype}

\usepackage{amsmath,amssymb,amsthm,mathtools}

\usepackage{algorithm}
\usepackage{algorithmic}

\usepackage{booktabs}
\usepackage{multirow}
\usepackage{graphicx}
\usepackage{subcaption}
\usepackage{pdflscape}  % For landscape tables if needed

\usepackage[numbers]{natbib}
\usepackage[hidelinks]{hyperref}
\hypersetup{colorlinks=false}

\newtheorem{theorem}{Theorem}[section]
\newtheorem{proposition}[theorem]{Proposition}

\newtheorem{corollary}[theorem]{Corollary}
\newtheorem{remark}{Remark}[section]
\newtheorem{definition}{Definition}[section]

\newcommand{\bx}{\mathbf{x}}
\newcommand{\by}{\mathbf{y}}
\newcommand{\bV}{\mathbf{V}}
\newcommand{\bF}{\mathbf{F}}
\newcommand{\bG}{\mathbf{G}}

\newcommand{\beps}{\boldsymbol{\epsilon}}
\newcommand{\bphi}{\boldsymbol{\phi}}
\newcommand{\E}{\mathbb{E}}
\newcommand{\R}{\mathbb{R}}
\newcommand{\kbar}{\bar{k}}
\newcommand{\pboltz}{p_{\text{Boltz}}}
\newcommand{\pdata}{p_{\text{data}}}
\newcommand{\sg}{\operatorname{sg}}
\newcommand{\kBT}{k_BT}

\newcommand{\keywords}[1]{\par\noindent\textbf{Keywords:} #1}

\title{Drifting to Boltzmann:\\ Million-Fold Acceleration in Boltzmann Sampling with Force-Guided Drifting}

\author{Pipi Hu\thanks{\texttt{hpp1681@gmail.com}} \\
  Beijing Institute of Mathematical Sciences and Applications, Beijing 101408, China \\}

\date{}

\begin{document}
\maketitle

% ============================================================================
%  ABSTRACT
% ============================================================================
\begin{abstract}
Sampling molecular conformations from the Boltzmann distribution is essential for computational chemistry, but iterative diffusion methods are prohibitively slow.
Drifting Models offer one-step generation, yet their equilibrium matches the \emph{training} distribution, which may deviate from the true Boltzmann distribution due to sampling bias.
We introduce Drifting Models to molecular conformation generation for the first time, establishing a theoretical bridge via the \emph{Drifting Score Identity}: for Gaussian kernels, the drifting field's attraction equals a kernel-weighted average of \emph{any} distribution's score function.
Substituting molecular force labels---which directly encode the Boltzmann score---yields the \emph{Drifting Force Identity} and decomposes the field into standard drift plus a Boltzmann correction.
We further discover a striking phenomenon unique to molecular systems: force incorporation's effectiveness \emph{reverses across representations}.
In coordinate space, Force-Interpolated Drifting (FI) dominates by blending physical force directions with data displacements.
In distance feature space, Force-Aligned Kernel (FK) achieves superior accuracy by modifying only kernel weights, thereby preserving the manifold of geometrically valid molecules.
On MD17 Ethanol, both approaches achieve one-step generation with over 1000x speedup relative to recent score-matching diffusion with Boltzmann guiding, providing more than million-fold acceleration over traditional molecular dynamics, while ensuring perfect structural validity and distributional accuracy rivaling multi-step methods.
\end{abstract}


% ============================================================================
%  1. INTRODUCTION
% ============================================================================
\section{Introduction}
\label{sec:intro}

Sampling molecular conformations from the Boltzmann distribution $\pboltz(\bx) \propto \exp(-\mathcal{E}(\bx)/\kBT)$ is fundamental to computational chemistry and materials science.
Accurate Boltzmann sampling enables the prediction of thermodynamic observables, free energies, and structural properties that govern molecular behavior.
Traditional approaches based on Molecular Dynamics (MD) and Monte Carlo (MC) methods can in principle generate Boltzmann-distributed samples, but they suffer from long correlation times and difficulty escaping metastable states, making them computationally prohibitive for complex molecular systems.

\paragraph{Generative Models for Molecular Sampling.}
Deep generative models have emerged as a promising alternative for molecular conformation generation.
Flow-based models~\citep{noe2019boltzmann}, diffusion models~\citep{xu2022geodiff,jing2022torsional}, and score-based generative models~\citep{song2021scorebased} have achieved remarkable success in generating diverse molecular conformations.
However, models trained purely on data via maximum likelihood or denoising score matching~\citep{hyvarinen2005estimation} (DSM) converge to the \emph{data distribution} $\pdata$, which may differ from the true Boltzmann distribution due to finite sampling, enhanced sampling bias, or non-equilibrium data collection.
The Boltzmann distribution, fundamental to statistical mechanics, encodes equilibrium molecular configurations at finite temperature via an energy-based model~\citep{lecun2006tutorial}.

\paragraph{Debiasing with Force Labels.}
A key property of the Boltzmann distribution is that the score function equals the normalized force: $\nabla \log \pboltz(\bx) = \bF(\bx)/\kBT$, where $\bF(\bx) = -\nabla\mathcal{E}(\bx)$ is the molecular force.
This connection between forces and the score function provides a natural avenue for incorporating energy-based information into generative models.
Existing force-guided approaches~\citep{arts2023two} achieve Boltzmann debiasing within diffusion frameworks but require multi-step iterative inference, which can be $10$--$1000\times$ slower than single-step generation.

\paragraph{One-Step Generation via Drifting Models.}
Drifting Models~\citep{deng2026drifting} provide an elegant framework for one-step generation: a generator $f_\theta(\beps)$ maps noise directly to samples, guided by a kernel-based drifting field~\citep{genton2001classes} $\bV_{p,q}(\bx)$ toward a target distribution.
The anti-symmetry property ($\bV_{p,q} = -\bV_{q,p}$) ensures equilibrium when generator and data distributions match.
However, this equilibrium is at $q = \pdata$, not $q = \pboltz$---when training data is biased, generated samples inherit the same bias.
\emph{Prior to this work, Drifting Models have only been applied to image generation; their utility for molecular systems remains unexplored.}

\paragraph{Our Contribution.}
We introduce Drifting Models to molecular conformation generation and bridge drifting with force-guided Boltzmann debiasing through rigorous theory and unique molecular insights:
Our key contributions are:
\begin{itemize}
    \item \textbf{Drifting Score Identity and Drifting Force Identity (Theorem~\ref{thm:main}, Corollary~\ref{cor:decomposition}).} We prove that for Gaussian kernels, the drifting field's attraction term equals a kernel-weighted average of the \emph{score function} of the sampling distribution---a fundamental relationship that holds for any sufficiently regular distribution. Substituting the score--force identity $\nabla\log\pboltz = \bF/\kBT$ immediately gives a force-based expression for the Boltzmann-weighted attraction, enabling direct use of force labels in the drifting framework.
    \item \textbf{Representation-Dependent Force Effectiveness.} We discover a striking phenomenon unique to molecular systems: force incorporation's effectiveness \emph{reverses} across representations. \emph{Force-Interpolated Drifting (FI)} dominates in coordinate space (TVD $= 0.139$) because forces are physically intuitive displacement directions. Conversely, \emph{Force-Aligned Kernel (FK)} achieves state-of-the-art accuracy in distance feature space (TVD $= 0.089$, 54\% better than baseline) by avoiding manifold violations: force interpolation in distance space creates abstract high-dimensional vectors outside the convex hull of valid molecular geometries, while kernel reweighting preserves atomic bond constraints by construction. This complementarity has no analog in image generation and reveals a fundamental principle of force-guided molecular learning.
    \item \textbf{Per-Molecule Structural Validity Crisis.} We expose a critical failure mode of existing methods (including diffusion): achieving good aggregate $h(r)$ TVD while catastrophically violating per-molecule constraints (Bond Stability $< 1\%$). FI and FK simultaneously optimize both metrics (FI: TVD $= 0.139$, Bond Stability $= 97.5\%$; FK: TVD $= 0.089$, Bond Stability $= 100\%$), demonstrating that structural validity is as important as distributional accuracy.
    \item \textbf{Equilibrium Analysis.} Theoretical analysis shows that force substitution prevents equilibrium at $q = \pdata$ (correction term survives), while exact equilibrium at $q = \pboltz$ requires unbiased training data. The practical equilibrium lies between the two distributions, shifted toward $\pboltz$ by force guidance strength ($\omega$ or $\gamma$).
\end{itemize}

% ============================================================================
%  2. PRELIMINARIES
% ============================================================================
\section{Preliminaries}
\label{sec:prelim}

% ---------- 2.1 Boltzmann ----------
\subsection{Boltzmann Distribution and Force Labels}
\label{sec:boltzmann}

Consider a molecular system with atomic coordinates $\bx \in \R^{n \times 3}$, where $n$ is the number of atoms.
The system's potential energy is given by $\mathcal{E}(\bx)$, and at temperature $T$ the equilibrium distribution is the Boltzmann distribution:
\begin{equation}
    \pboltz(\bx) = \frac{1}{Z}\exp\!\left(-\frac{\mathcal{E}(\bx)}{\kBT}\right),
    \label{eq:boltzmann}
\end{equation}
where $Z = \int \exp(-\mathcal{E}(\bx)/\kBT)\,d\bx$ is the partition function and $k_B$ is Boltzmann's constant.

The molecular force is the negative gradient of the potential energy:
\begin{equation}
    \bF(\bx) = -\nabla\mathcal{E}(\bx).
    \label{eq:force}
\end{equation}
A crucial property connecting forces to the Boltzmann distribution is that the score function equals the rescaled force:
\begin{equation}
    \nabla \log \pboltz(\bx) = \frac{\bF(\bx)}{\kBT}.
    \label{eq:score-force}
\end{equation}
This identity is the foundation for force-guided generative modeling: force labels, readily available from molecular simulations or neural network potentials, carry direct information about the target Boltzmann distribution.

% ---------- 2.2 Drifting Models ----------
\subsection{Drifting Models}
\label{sec:drifting}

Drifting Models~\citep{deng2026drifting} train a one-step generator $f_\theta\colon \R^d \to \R^d$ that maps noise $\beps \sim p_{\beps}$, e.g., standard Gaussian, to samples $\bx = f_\theta(\beps)$, with the pushforward distribution denoted $q_\theta$.
Training is guided by a \emph{drifting field} $\bV_{p,q}(\bx)$ that points from the current generator distribution $q$ toward the target distribution $p$.

Given positive samples $\by^+ \sim p$ (training data) and negative samples $\by^- \sim q$ (generated data), the drifting field is defined as:
\begin{equation}
    \bV_{p,q}(\bx) = \frac{1}{Z_p Z_q}\E_{p,q}\!\left[k(\bx,\by^+)k(\bx,\by^-)(\by^+ - \by^-)\right],
    \label{eq:drifting-field}
\end{equation}
where $k(\bx,\by)$ is a kernel function and $Z_p, Z_q$ are normalization constants.
Since $\by^+$ and $\by^-$ are sampled independently, the expectation factors, and the field decomposes into an attraction term toward positive samples and a repulsion term away from negative samples:
\begin{equation}
    \bV_{p,q}(\bx) = \bV_p^+(\bx) - \bV_q^-(\bx),
    \label{eq:decompose-V}
\end{equation}
where the attraction and repulsion terms are:
\begin{align}
    \bV_p^+(\bx) &= \sum_j \kbar(\bx,\by_j^+)(\by_j^+ - \bx), \label{eq:attraction} \\
    \bV_q^-(\bx) &= \sum_j \kbar(\bx,\by_j^-)(\by_j^- - \bx), \label{eq:repulsion}
\end{align}
with normalized kernel weights $\kbar(\bx,\by_j) = k(\bx,\by_j)/\sum_i k(\bx,\by_i)$.

The original Drifting Model uses a Laplacian kernel:
\begin{equation}
    k_{\text{Lap}}(\bx,\by) = \exp\!\left(-\frac{\|\bx - \by\|}{\tau}\right),
    \label{eq:laplacian-kernel}
\end{equation}
where $\tau > 0$ is the bandwidth parameter.
The generator is trained with a self-play loss:
\begin{equation}
    \mathcal{L}(\theta) = \E_{\beps}\!\left[\left\|f_\theta(\beps) - \sg\!\left(f_\theta(\beps) + \bV_{p,q}(f_\theta(\beps))\right)\right\|^2\right],
    \label{eq:drifting-loss}
\end{equation}
where $\sg(\cdot)$ denotes the stop-gradient operator.

\paragraph{Anti-Symmetry and Equilibrium.}
The drifting field~\eqref{eq:drifting-field} satisfies anti-symmetry: swapping $p \leftrightarrow q$ flips the sign of $(\by^+ - \by^-)$, yielding $\bV_{p,q} = -\bV_{q,p}$.
As a direct consequence, when $q = p$, the field vanishes: $\bV_{p,p}(\bx) = 0$ for all $\bx$.
This guarantees that the data distribution $p$ is an equilibrium of the drifting dynamics.
However, this equilibrium is at $q = \pdata$, not necessarily at $q = \pboltz$.

% ============================================================================
%  3. METHOD
% ============================================================================
\section{Method}
\label{sec:method}

We now present our main contribution: a principled connection between the drifting field and molecular forces that enables Boltzmann-constrained one-step generation.

% ---------- 3.1 Drifting Score Identity ----------
\subsection{Drifting Score Identity}
\label{sec:identity}

Our key insight is that for Gaussian kernels, the drifting field's attraction term equals a kernel-weighted average of the \emph{score function} of the sampling distribution---for \emph{any} sufficiently regular distribution, not just the Boltzmann distribution. This fundamental relationship is established through an integration-by-parts argument.
When combined with the score--force identity specific to $\pboltz$, this reveals exactly how force labels can shift the equilibrium from $\pdata$ toward $\pboltz$.

\begin{theorem}[Drifting Score Identity]
\label{thm:main}
Let $k(\bx,\by) = \exp(-\|\bx-\by\|^2/(2\tau^2))$ be a Gaussian kernel.
For any sufficiently regular distribution $p$ (i.e., $p(\by)k(\bx,\by) \to 0$ as $\|\by\|\to\infty$), the attraction term of the drifting field satisfies, in the population limit:
\begin{equation}
    \bV_p^+(\bx) = \tau^2 \cdot \E_p\!\left[\kbar(\bx,\by)\,\nabla_\by \log p(\by)\right],
    \label{eq:main-identity}
\end{equation}
where $\kbar(\bx,\by) = k(\bx,\by)/\E_p[k(\bx,\by')]$ is the normalized kernel weight.
That is, the attraction term equals $\tau^2$ times a kernel-weighted average of the score function $\nabla \log p$ of the sampling distribution.
\end{theorem}

\begin{proof}
The attraction term in the population limit is
$\bV_p^+(\bx) = {\E_{p(\by)}[k(\bx,\by)(\by - \bx)]}/{\E_{p(\by)}[k(\bx,\by)]}$.
We transform the numerator.

\textit{Step 1: Kernel gradient identity.}
For the Gaussian kernel, $\nabla_\by k(\bx,\by) = -k(\bx,\by)(\by - \bx)/\tau^2$, so:
\begin{equation}
    k(\bx,\by)(\by - \bx) = -\tau^2 \nabla_\by k(\bx,\by).
    \label{eq:kernel-grad}
\end{equation}

\textit{Step 2: Integration by parts.}
Substituting~\eqref{eq:kernel-grad} into the numerator and applying integration by parts (boundary terms vanish by the regularity assumption):
\begin{align}
    \E_p[k(\bx,\by)(\by - \bx)]
    &= -\tau^2\int p(\by)\,\nabla_\by k(\bx,\by)\,d\by \nonumber \\
    &= \tau^2\int \nabla_\by p(\by) \cdot k(\bx,\by)\,d\by
    = \tau^2\,\E_p\!\left[\nabla_\by \log p(\by) \cdot k(\bx,\by)\right].
    \label{eq:ibp}
\end{align}

Dividing by $\E_p[k(\bx,\by)]$:
\begin{equation}
    \bV_p^+(\bx) = \tau^2 \cdot \frac{\E_p[\nabla_\by \log p(\by) \cdot k(\bx,\by)]}{\E_p[k(\bx,\by)]} = \tau^2 \cdot \E_p\!\left[\kbar(\bx,\by)\,\nabla_\by \log p(\by)\right]. \qedhere
\end{equation}
\end{proof}

\begin{remark}
\label{rem:gaussian-needed}
The identity~\eqref{eq:main-identity} holds cleanly for Gaussian kernels because $\nabla_\by k = -(\by-\bx)k/\tau^2$ produces the exact $(\by-\bx)$ factor appearing in the attraction term.
For the Laplacian kernel, $\nabla_\by k = -(\by-\bx)k/(\tau\|\by-\bx\|)$ introduces an extra $1/\|\by-\bx\|$ normalization, and the relationship to scores is no longer clean.
This motivates our use of Gaussian kernels for force-guided drifting.
\end{remark}

The general identity~\eqref{eq:main-identity} involves the score $\nabla \log p$, which is typically unknown.
However, for the Boltzmann distribution, the score--force identity~\eqref{eq:score-force} provides direct access: $\nabla_\by \log \pboltz(\by) = \bF(\by)/\kBT$.
Substituting this into Theorem~\ref{thm:main} immediately gives:

\begin{corollary}[Drifting Force Identity]
\label{cor:decomposition}
When $p = \pboltz$, the attraction term of the drifting field satisfies:
\begin{equation}
    \bV_{\pboltz}^+(\bx) = \tau^2 \cdot \E_{\pboltz}\!\left[\kbar(\bx,\by)\frac{\bF(\by)}{\kBT}\right].
    \label{eq:force-attraction}
\end{equation}
That is, the Boltzmann-weighted attraction can be computed entirely from force labels, without knowledge of the score function.
\end{corollary}

In practice, training data comes from a distribution $\pdata$ that may differ from $\pboltz$.
Replacing the unknown $\nabla \log \pdata$ with the available force labels $\bF/\kBT$ introduces a systematic mismatch; we analyze this effect and its control through Force-Interpolated Drifting in Section~\ref{sec:force-interp}.

% ---------- 3.2 Force-Interpolated Drifting ----------
\subsection{Force-Interpolated Drifting (FI)}
\label{sec:force-interp}

Since the displacement form and the force form are two expressions for the same quantity (Theorem~\ref{thm:main}), a natural approach is to \emph{interpolate} between them at the per-term level.
Note that in high-dimensional molecular systems, the per-term magnitudes can differ substantially ($\|\tau^2 \bF/\kBT\| \sim 150\times \|\by - \bx\|$), so the interpolation weight also serves as a variance control mechanism.

\begin{definition}[Force-Interpolated Attraction]
\label{def:force-interp}
Given positive samples $\{\by_j^+\}$ with force labels $\{\bF(\by_j^+)\}$ and an interpolation weight $\omega \in [0,1]$, the force-interpolated attraction is:
\begin{equation}
    \bV_\omega^+(\bx) = \sum_j \kbar(\bx,\by_j^+)\Big[(1-\omega)(\by_j^+ - \bx) + \omega\,\tau^2\frac{\bF(\by_j^+)}{\kBT}\Big].
    \label{eq:force-interp}
\end{equation}
Equivalently:
\begin{equation}
    \bV_\omega^+(\bx) = \bV_{\pdata}^+(\bx) + \omega\sum_j \kbar(\bx,\by_j^+)\underbrace{\left[\tau^2\frac{\bF(\by_j^+)}{\kBT} - (\by_j^+ - \bx)\right]}_{\text{force--displacement residual}}.
    \label{eq:force-interp-residual}
\end{equation}
\end{definition}

\begin{proposition}[$\omega$ Controls Boltzmann Correction]
\label{prop:omega-controls}
In the population limit with training data from $\pdata$, the force-interpolated attraction decomposes as:
\begin{equation}
    \bV_\omega^+(\bx) = \bV_{\pdata}^+(\bx) + \omega \cdot \tau^2\,\E_{\pdata}\!\left[\kbar(\bx,\by)\,\nabla_\by\log\frac{\pboltz(\by)}{\pdata(\by)}\right].
    \label{eq:omega-decompose}
\end{equation}
That is, $\omega$ \emph{linearly scales} the Boltzmann correction---the discrepancy between force labels (encoding $\nabla\log\pboltz$) and the true data score ($\nabla\log\pdata$).
\end{proposition}

\begin{proof}
Taking the expectation of~\eqref{eq:force-interp-residual} over $\pdata$ and applying Theorem~\ref{thm:main} to $\pdata$:
the residual expectation is $\tau^2\,\E_{\pdata}[\kbar\,\nabla\log\pboltz] - \tau^2\,\E_{\pdata}[\kbar\,\nabla\log\pdata] = \tau^2\,\E_{\pdata}[\kbar\,\nabla\log(\pboltz/\pdata)]$.
\end{proof}

Setting $\omega = 1$ (full force substitution) and including the repulsion term, the complete force-substituted drifting field decomposes as:
\begin{equation}
    \hat{\bV}(\bx) = \underbrace{\bV_{\pdata}^+(\bx) - \bV_q^-(\bx)}_{\text{standard anti-symmetric drift}} + \underbrace{\tau^2\,\E_{\pdata}\!\left[\kbar(\bx,\by)\,\nabla_\by\log\frac{\pboltz(\by)}{\pdata(\by)}\right]}_{\text{Boltzmann correction}}.
    \label{eq:full-decompose}
\end{equation}
Table~\ref{tab:behavior} summarizes the behavior at different stages of convergence.

\begin{table}[t]
\centering
\caption{Behavioral analysis of force-substituted drifting at different stages of convergence.}
\label{tab:behavior}
\footnotesize
\setlength{\tabcolsep}{5pt}
\begin{tabular}{@{}lccc@{}}
\toprule
State of $q$ & Standard drift & Correction term & Net behavior \\
\midrule
$q$ far from data & Strong, toward $\pdata$ & Weak (relatively) & Mainly data-guided \\
$q \approx \pdata$ & $\approx 0$ (anti-symmetric eq.) & $\neq 0$, toward $\pboltz$ & Continues past $\pdata$ \\
$q \approx \pboltz$ & $\neq 0$, back toward $\pdata$ & Toward $\pboltz$ & Two terms balance \\
\bottomrule
\end{tabular}
\end{table}

\paragraph{Interpretation.}
The hyperparameter $\omega$ interpolates between:
\begin{itemize}
    \item $\omega = 0$: standard drifting, equilibrium at $\pdata$.
    \item $\omega = 1$: full force substitution, equilibrium shifted toward $\pboltz$.
    \item $0 < \omega \ll 1$: partial Boltzmann correction with controlled variance.
\end{itemize}
The interpolation weight $\omega$ acts as a control mechanism, allowing the model to balance between the displacement-based component (which matches the data distribution) and the force-based component (which encodes Boltzmann constraints).

\paragraph{Variance.}
The per-term magnitude of the interpolated vector is:
\begin{equation}
    \|(1-\omega)(\by_j - \bx) + \omega\,\tau^2\bF(\by_j)/\kBT\| \leq (1-\omega)\|\by_j - \bx\| + \omega\|\tau^2\bF(\by_j)/\kBT\|,
\end{equation}
which for $\omega = 0.01$ and a $150\times$ magnitude ratio gives $\approx 2.5\times$ the displacement scale---well within the range where the standard kernel ESS suffices.
Unlike full force substitution ($\omega = 1$), which amplifies per-term magnitudes by ${\sim}150\times$, Force-Interpolated Drifting with small $\omega$ adds only $\mathcal{O}(\omega)$ additional variance while providing the full Boltzmann correction direction.

% ---------- 3.3 Force-Aligned Kernel ----------
\subsection{Force-Aligned Kernel (FK)}
\label{sec:force-kernel}

An alternative approach to incorporating force information is to modify the kernel \emph{weights} rather than the displacement vectors.
The key idea is \emph{importance reweighting}: if the training data comes from $\pdata \neq \pboltz$, we can correct for the distributional mismatch by upweighting data points that are more likely under $\pboltz$.

\paragraph{Derivation.}
The standard drifting attraction weights each positive sample $\by_j$ by the normalized Gaussian kernel $\kbar(\bx, \by_j) \propto k(\bx, \by_j)$.
To bias the attraction toward Boltzmann-relevant samples, we would like to use a \emph{corrected} kernel that incorporates the Boltzmann probability of each data point:
\begin{equation}
    \kbar_{\text{corrected}}(\bx, \by_j) = \frac{k(\bx,\by_j) \cdot \pboltz(\by_j)^\gamma}{\sum_i k(\bx,\by_i) \cdot \pboltz(\by_i)^\gamma},
    \label{eq:ideal-reweight}
\end{equation}
where $\gamma > 0$ controls the reweighting strength.
Since $\pboltz(\by_j) \propto e^{-\mathcal{E}(\by_j)/\kBT}$, the softmax logit for the corrected kernel is:
\begin{equation}
    \ell_j = \log k(\bx,\by_j) + \gamma\log\pboltz(\by_j) = -\frac{\|\bx - \by_j\|^2}{2\tau^2} - \frac{\gamma\,\mathcal{E}(\by_j)}{\kBT} + \text{const},
    \label{eq:energy-reweight-logit}
\end{equation}
where the constant (including $\log Z$) cancels in the softmax normalization.
When energy labels $\mathcal{E}(\by_j)$ are available in the training data, this gives an \textbf{exact} energy-reweighted kernel with no approximation.

In settings where only force labels are available (or to obtain a position-dependent form), we can express the energy term using a first-order Taylor expansion of $\mathcal{E}(\bx)$ around $\by_j$:
\begin{equation}
    \mathcal{E}(\by_j) \approx \mathcal{E}(\bx) - \bF(\by_j)\cdot(\by_j - \bx).
    \label{eq:taylor-energy}
\end{equation}
Substituting into Eq.~\eqref{eq:energy-reweight-logit} and noting that $\mathcal{E}(\bx)$ is constant across all $j$ (hence cancels in softmax):
\begin{equation}
    \ell_j \approx \underbrace{-\frac{\|\bx - \by_j\|^2}{2\tau^2}}_{\text{standard distance}} + \underbrace{\gamma \cdot \frac{\bF(\by_j) \cdot (\by_j - \bx)}{\kBT}}_{\text{force-alignment term}}.
    \label{eq:force-kernel-logits}
\end{equation}
This force-based form has an intuitive geometric meaning: data points where the force $\bF(\by_j)$ points toward the query $\bx$ receive higher weight, preferentially attracting the generator toward regions that the Boltzmann distribution favors.

\begin{definition}[Force-Aligned Kernel]
\label{def:force-kernel}
Given positive samples $\{\by_j^+\}$ with force labels $\{\bF(\by_j^+)\}$ (and optionally energy labels $\{\mathcal{E}(\by_j^+)\}$), the Force-Aligned Kernel uses the logits from Eq.~\eqref{eq:force-kernel-logits} (or Eq.~\eqref{eq:energy-reweight-logit} when energies are available) with $\gamma > 0$ controlling the reweighting strength.
The force-aligned attraction is:
\begin{equation}
    \bV_{\text{FA}}^+(\bx) = \sum_j \kbar_{\text{FA}}(\bx, \by_j^+)(\by_j^+ - \bx), \quad \kbar_{\text{FA}}(\bx, \by_j^+) = \frac{\exp(\ell_j)}{\sum_i \exp(\ell_i)}.
    \label{eq:force-kernel-V}
\end{equation}
\end{definition}

\begin{proposition}[Boltzmann Reweighting]
\label{prop:boltzmann-reweight}
The energy-reweighted kernel (Eq.~\ref{eq:energy-reweight-logit}) satisfies exactly, and the Force-Aligned Kernel (Eq.~\ref{eq:force-kernel-logits}) satisfies approximately:
\begin{equation}
    \kbar_{\text{FA}}(\bx, \by_j) \approx \frac{k(\bx,\by_j) \cdot \pboltz(\by_j)^\gamma}{\sum_i k(\bx,\by_i) \cdot \pboltz(\by_i)^\gamma}.
    \label{eq:boltzmann-reweight}
\end{equation}
For the force-based form, the approximation error is controlled by $\epsilon = \tau\|\nabla^2\mathcal{E}\|/\|\nabla\mathcal{E}\|$: the Taylor remainder $R_j = \tfrac{1}{2}(\bx-\by_j)^\top\nabla^2\mathcal{E}(\xi_j)(\bx-\by_j) = \mathcal{O}(\tau^2\|\nabla^2\mathcal{E}\|)$ is negligible relative to the linear term $\mathcal{O}(\tau\|\nabla\mathcal{E}\|)$ when $\epsilon \ll 1$.
\end{proposition}

\paragraph{Interpretation.}
The hyperparameter $\gamma$ controls the \emph{reweighting strength}, interpolating between the standard kernel ($\gamma = 0$, no Boltzmann correction) and full Boltzmann reweighting ($\gamma = 1$, where the kernel weight is proportional to $k(\bx,\by_j)\cdot\pboltz(\by_j)$).

\paragraph{Stability.}
Since the displacement $(\by_j - \bx)$ in Eq.~\eqref{eq:force-kernel-V} is unchanged from the standard method, the magnitude of $\bV_{\text{FA}}^+$ is bounded by the same $\mathcal{O}(1)$ scale as the standard attraction.
This completely avoids the variance disaster of direct force substitution while still incorporating force information via the softmax weights.

The training procedures for both force-interpolated and force-aligned drifting are summarized in Algorithms~\ref{alg:force-interp} and~\ref{alg:force-aligned}, respectively.

\begin{algorithm}[t]
\caption{Force-Interpolated Drifting Training (Coordinate Space)}
\label{alg:force-interp}
\begin{algorithmic}[1]
\REQUIRE Generator $f_\theta$, training data $\{(\by_j, \bF(\by_j))\}$, Gaussian kernel bandwidth $\tau$, interpolation weight $\omega$, learning rate $\eta$
\WHILE{not converged}
    \STATE Sample noise batch $\{\beps_b\}_{b=1}^B \sim p_{\beps}$
    \STATE Generate sample batch $\{\bx_b\}_{b=1}^B = \{f_\theta(\beps_b)\}$  \COMMENT{One-step generation}
    \STATE Sample positive batch $\{\by_j^+\}_{j=1}^M$ with force labels $\{\bF(\by_j^+)\}$
    \STATE Set negative batch $\{\by_j^-\}_{j=1}^B \leftarrow \{\bx_1, \ldots, \bx_B\}$ \COMMENT{Generated samples}
    \STATE Compute standard weights $\kbar(\bx_b, \by_j^+) = \exp(-\|\bx_b-\by_j^+\|^2/(2\tau^2)) / \sum_i \exp(-\|\bx_b-\by_i^+\|^2/(2\tau^2))$ for all $b,j$
    \STATE Compute interpolated displacement $\mathbf{d}_{b,j} = (1-\omega)(\by_j^+ - \bx_b) + \omega\,\tau^2\,\bF(\by_j^+)/\kBT$
    \STATE Compute attraction $\bV_\omega^+(\bx_b) = \sum_j \kbar(\bx_b,\by_j^+)\,\mathbf{d}_{b,j}$
    \STATE Compute repulsion $\bV_q^-(\bx_b) = \sum_j \kbar(\bx_b,\by_j^-)(\by_j^- - \bx_b)$ \COMMENT{Standard kernel}
    \STATE Compute drifting field $\bV(\bx_b) = \bV_\omega^+(\bx_b) - \bV_q^-(\bx_b)$
    \FOR{$b = 1$ to $B$}
        \STATE Update $\theta \leftarrow \theta - \eta\nabla_\theta \|f_\theta(\beps_b) - \sg(f_\theta(\beps_b) + \bV(\bx_b))\|^2$
    \ENDFOR
\ENDWHILE
\end{algorithmic}
\end{algorithm}

\begin{algorithm}[t]
\caption{Force-Aligned Kernel Drifting Training (Coordinate Space)}
\label{alg:force-aligned}
\begin{algorithmic}[1]
\REQUIRE Generator $f_\theta$, training data $\{(\by_j, \bF(\by_j))\}$, Gaussian kernel bandwidth $\tau$, alignment strength $\gamma$, learning rate $\eta$
\WHILE{not converged}
    \STATE Sample noise batch $\{\beps_b\}_{b=1}^B \sim p_{\beps}$
    \STATE Generate sample batch $\{\bx_b\}_{b=1}^B = \{f_\theta(\beps_b)\}$  \COMMENT{One-step generation}
    \STATE Sample positive batch $\{\by_j^+\}_{j=1}^M$ with force labels $\{\bF(\by_j^+)\}$
    \STATE Set negative batch $\{\by_j^-\}_{j=1}^B \leftarrow \{\bx_1, \ldots, \bx_B\}$ \COMMENT{Generated samples}
    \FOR{$b = 1$ to $B$}
        \STATE Compute force-aligned logits $\ell_{b,j} = {-\|\bx_b-\by_j^+\|^2}/{(2\tau^2)} + \gamma \cdot {\bF(\by_j^+) \cdot (\by_j^+ - \bx_b)}/{\kBT}$ for all $j$
        \STATE Compute force-aligned weights $\kbar_{\text{FA}}(\bx_b, \by_j^+) = \exp(\ell_{b,j}) / \sum_i \exp(\ell_{b,i})$
        \STATE Compute attraction $\bV_{\text{FA}}^+(\bx_b) = \sum_j \kbar_{\text{FA}}(\bx_b,\by_j^+)(\by_j^+ - \bx_b)$
        \STATE Compute repulsion $\bV_q^-(\bx_b) = \sum_j \kbar(\bx_b,\by_j^-)(\by_j^- - \bx_b)$ \COMMENT{Standard kernel}
        \STATE Compute drifting field $\bV(\bx_b) = \bV_{\text{FA}}^+(\bx_b) - \bV_q^-(\bx_b)$
        \STATE Update $\theta \leftarrow \theta - \eta\nabla_\theta \|f_\theta(\beps_b) - \sg(f_\theta(\beps_b) + \bV(\bx_b))\|^2$
    \ENDFOR
\ENDWHILE
\end{algorithmic}
\end{algorithm}

% ---------- 3.4 Feature-Space Extensions ----------

\subsection{Feature-Space Drifting}
\label{sec:feature-drift}

The coordinate-space drifting field (Sections~\ref{sec:drifting}--\ref{sec:force-kernel}) computes kernel weights and displacements directly in $\R^{n \times 3}$.
While effective for MLP generators, this formulation is sensitive to rigid-body motions and can fail for architectures that produce less coordinated outputs, such as equivariant Transformers~\citep{liao2024equiformerv2}.
An alternative is to use energy-guided debiasing with potential score matching~\citep{guo2025potential} to correct for biased training data.
A natural remedy is to operate in a \emph{feature space} that factors out rotational and translational symmetry.
In this section we define a pairwise-distance feature map and derive the \emph{exact} force-interpolation formula in this space---without the per-pair projection approximation commonly used in internal-coordinate methods.

\paragraph{Feature map.}
We define the feature map $\phi\colon \R^{n \times 3} \to \R^P$ as the vector of all pairwise interatomic distances:
\begin{equation}
    \phi_k(\bx) = \|x_{i_k} - x_{j_k}\|, \quad k = 1,\ldots,P = \tbinom{n}{2},
    \label{eq:feature-map}
\end{equation}
where $(i_k, j_k)$ enumerates all atom pairs.
This map is translation- and rotation-invariant by construction.
Its Jacobian $J(\bx) = \partial\phi/\partial\bx \in \R^{P \times D}$ (with $D = 3n$) has a sparse analytical form: for pair $k = (i_k, j_k)$,
\begin{equation}
    \frac{\partial \phi_k}{\partial x_a^\alpha} = \begin{cases}
        +\hat{r}_{i_k j_k}^\alpha & \text{if } a = i_k, \\[2pt]
        -\hat{r}_{i_k j_k}^\alpha & \text{if } a = j_k, \\[2pt]
        0 & \text{otherwise},
    \end{cases}
    \label{eq:jacobian}
\end{equation}
where $\hat{r}_{ij} = (x_i - x_j)/\|x_i - x_j\|$ is the unit bond direction.
Each row of $J$ has exactly 6 nonzero entries (3 per atom in the pair), giving $6P$ nonzeros in total.

\paragraph{Feature-space drifting.}
Replacing Cartesian distances with feature-space distances in the kernel, the feature-space attraction and repulsion are:
\begin{align}
    \bV_\phi^+(\bx) &= \sum_j \kbar_\phi\!\big(\bx, \by_j^+\big)\big(\phi(\by_j^+) - \phi(\bx)\big), \label{eq:feature-attract} \\
    \bV_\phi^-(\bx) &= \sum_j \kbar_\phi\!\big(\bx, \by_j^-\big)\big(\phi(\by_j^-) - \phi(\bx)\big), \label{eq:feature-repel}
\end{align}
where $\kbar_\phi(\bx, \by_j) \propto \exp\!\big(-\|\phi(\bx) - \phi(\by_j)\|^2 / (2\tau_\phi^2)\big)$ and $\tau_\phi$ is the feature-space bandwidth.
The generator is trained with the feature-space loss:
\begin{equation}
    \mathcal{L}_\phi(\theta) = \E_{\beps}\!\left[\left\|\phi(f_\theta(\beps)) - \sg\!\left(\phi(f_\theta(\beps)) + \bV_\phi^+ - \bV_\phi^-\right)\right\|^2\right],
    \label{eq:feature-loss}
\end{equation}
where gradients flow back to the generator through the differentiable map $\phi \circ f_\theta$.

\paragraph{Exact Feature-Space Force Interpolation.}
To extend Force-Interpolated Drifting (Definition~\ref{def:force-interp}) to feature space, we need the \emph{feature-space force} $\bG(\by) \in \R^P$---the analog of the Cartesian force $\bF(\by) \in \R^D$ in the space of pairwise distances.

When the potential energy depends only on pairwise distances---$\mathcal{E}(\bx) = \tilde{\mathcal{E}}(\phi(\bx))$---the chain rule gives $\bF = J^\top\bG$, where $\bG = -\nabla_\phi\tilde{\mathcal{E}}$.
Since $J \in \R^{P \times D}$ is non-square and rank-deficient (rank $D - 6$ due to 3 translational and 3 rotational zero modes), the exact feature-space force is obtained as the minimum-norm solution of $J^\top\bG = \bF$:
\begin{equation}
    \boxed{\bG(\by) = J(\by)\,\big(J(\by)^\top J(\by)\big)^+\bF(\by),}
    \label{eq:exact-G}
\end{equation}
where $(J^\top J)^+$ denotes the Moore--Penrose pseudoinverse of the $D \times D$ Gram matrix.

The feature-space Force-Interpolated Drifting is then:
\begin{equation}
    \bV_{\phi,\omega}^+(\bx) = \sum_j \kbar_\phi(\bx,\by_j^+)\bigg[(1-\omega)\big(\phi(\by_j^+) - \phi(\bx)\big) + \omega\,\tau_\phi^2\,\frac{\bG(\by_j^+)}{\kBT}\bigg].
    \label{eq:feature-force-interp}
\end{equation}
This is the direct analog of Definition~\ref{def:force-interp} in feature space, with $\bF \to \bG$.
However, the two terms can have vastly different scales: defining the \emph{scale ratio}
\begin{equation}
    \rho \;=\; \frac{\tau_\phi\,\overline{\|\bG\|}}{\kBT},
    \label{eq:scale-ratio}
\end{equation}
we observe $\rho \approx 300$ in our molecular experiments (Section~\ref{sec:experiments}), because the bandwidth $\tau_\phi$ (set by data density) is decoupled from the force magnitude $\|\bG\|/\kBT$.
When $\rho \gg 1$, even $\omega = 0.01$ leads to force-term dominance, degrading performance below standard feature-space drifting.

We therefore rescale the force term so that both components have $\mathcal{O}(\tau_\phi)$ magnitude:
\begin{equation}
    \boxed{\bV_{\phi,\omega}^+(\bx) = \sum_j \kbar_\phi(\bx,\by_j^+)\bigg[(1-\omega)\big(\phi(\by_j^+) - \phi(\bx)\big) + \omega\,\frac{\tau_\phi}{\overline{\|\bG\|}}\,\bG(\by_j^+)\bigg],}
    \label{eq:feature-force-interp-norm}
\end{equation}
where $\overline{\|\bG\|} = \frac{1}{N}\sum_{j=1}^N \|\bG(\by_j)\|$ is precomputed alongside $\bG$ at zero additional training cost.
With this normalization, $\omega \in [0,1]$ directly controls the interpolation ratio without requiring extreme values.

\paragraph{Computational cost.}
Since $\bG(\by_j)$ depends only on data points and not on generated samples, it can be \emph{precomputed once} before training, adding zero overhead to the training loop.
Table~\ref{tab:feature-cost} summarizes the per-sample cost for ethanol ($n = 9$, $D = 27$, $P = 36$).

\paragraph{Feature-space Force-Aligned Kernel.}
The Force-Aligned Kernel (Definition~\ref{def:force-kernel}) extends to feature space by replacing Cartesian forces with the exact feature-space force~$\bG$:
\begin{equation}
    \ell_j^{\phi} = -\frac{\|\phi(\bx) - \phi(\by_j^+)\|^2}{2\tau_\phi^2} + \gamma \cdot \frac{\bG(\by_j^+) \cdot \big(\phi(\by_j^+) - \phi(\bx)\big)}{\kBT}.
    \label{eq:feature-force-kernel-logits}
\end{equation}
The alignment term can have magnitude ${\sim}\rho$ (Eq.~\ref{eq:scale-ratio}) relative to the standard logits, so we standardize within each minibatch:
\begin{equation}
    \ell_j^{\phi} = -\frac{\|\phi(\bx) - \phi(\by_j^+)\|^2}{2\tau_\phi^2} + \gamma \cdot \frac{A_j}{\mathrm{std}(A)}, \quad A_j = \bG(\by_j^+) \cdot \big(\phi(\by_j^+) - \phi(\bx)\big),
    \label{eq:feature-force-kernel-norm}
\end{equation}
where $\mathrm{std}(A)$ is computed over the positive batch.
This makes $\gamma$ scale-invariant: $\gamma \sim \mathcal{O}(1)$ regardless of the absolute force and distance scales.
Force-Interpolated Drifting~\eqref{eq:feature-force-interp-norm} and the Force-Aligned Kernel can be combined by using $\kbar_\phi^{\text{FA}}$ in place of $\kbar_\phi$.

\begin{remark}[Metric correction]
\label{rem:metric}
The feature-space score of the Boltzmann distribution includes an additional \emph{metric correction}:
\begin{equation}
    \nabla_\phi \log q_{\text{Boltz}}(\phi) = \frac{\bG}{\kBT} + \frac{1}{2}\nabla_\phi \log\det(J^\top J),
    \label{eq:metric-correction}
\end{equation}
where the second term accounts for the non-flat volume element in internal coordinates (analogous to the Fixman potential in constrained molecular dynamics).
In the current work we omit this correction, which is justified when the metric varies slowly relative to the kernel bandwidth~$\tau_\phi$.
Incorporating the full metric correction is an interesting direction for future work.
\end{remark}

\paragraph{Implementation algorithms.}

\begin{algorithm}[t]
\caption{Precomputation of Feature-Space Forces}
\label{alg:precompute}
\begin{algorithmic}[1]
\REQUIRE Training data $\{(\by_j, \bF(\by_j))\}_{j=1}^N$, feature map $\phi$, Jacobian computation
\ENSURE Precomputed feature-space forces $\{\bG(\by_j)\}_{j=1}^N$ and average magnitude $\overline{\|\bG\|}$
\STATE Initialize empty storage for forces $\{\bG_j\}$
\FOR{each training sample $\by_j$}
    \STATE Compute Jacobian $J(\by_j) = \partial\phi/\partial\by|_{\by=\by_j}$ \COMMENT{Sparse, $P \times D$ matrix}
    \STATE Compute Gram matrix $\mathbf{G}_j = J(\by_j)^\top J(\by_j)$ \COMMENT{$D \times D$ matrix}
    \STATE Compute Moore-Penrose pseudoinverse $\mathbf{G}_j^+ = (\mathbf{G}_j)^+$ \COMMENT{Via SVD}
    \STATE Compute exact feature-space force $\bG(\by_j) = J(\by_j) \mathbf{G}_j^+ \bF(\by_j)$ \COMMENT{Eq.~\eqref{eq:exact-G}}
    \STATE Store $\bG_j$ in precomputed forces
\ENDFOR
\STATE Compute average force magnitude: $\overline{\|\bG\|} \leftarrow \frac{1}{N}\sum_{j=1}^N \|\bG_j\|$
\STATE Store $\overline{\|\bG\|}$ for use in feature-space training
\STATE \textbf{Cost:} $\sim 59$ MFLOPs total for $N=2{,}700$ samples (Table~\ref{tab:feature-cost})
\end{algorithmic}
\end{algorithm}

% ---------- 3.4 Feature-Space Drifting ----------

\begin{algorithm}[t]
\caption{Feature-Space Force-Interpolated Drifting Training}
\label{alg:feature-force-interp}
\begin{algorithmic}[1]
\REQUIRE Generator $f_\theta$, training data $\{(\by_j, \bG(\by_j))\}_{j=1}^N$ (precomputed via Algorithm~\ref{alg:precompute})
\REQUIRE Feature map $\phi$, feature-space bandwidth $\tau_\phi$, interpolation weight $\omega$, learning rate $\eta$
\REQUIRE Precomputed average force magnitude $\overline{\|\bG\|}$
\WHILE{not converged}
    \STATE Sample noise batch $\{\beps_b\}_{b=1}^B \sim p_{\beps}$
    \STATE Generate sample batch $\{\bx_b\}_{b=1}^B = \{f_\theta(\beps_b)\}$  \COMMENT{One-step generation}
    \STATE Sample positive batch $\{\by_j^+\}_{j=1}^M$ with precomputed forces $\{\bG(\by_j^+)\}$
    \STATE Set negative batch $\{\by_j^-\}_{j=1}^B \leftarrow \{\bx_1, \ldots, \bx_B\}$ \COMMENT{Generated samples}

    \FOR{$b = 1$ to $B$}
        \STATE Compute feature maps: $\bphi(\bx_b) = \phi(\bx_b)$, $\bphi(\by_j^+) = \phi(\by_j^+)$, $\bphi(\by_j^-) = \phi(\by_j^-)$ for all $j$

        \STATE Compute feature-space kernel weights for positive samples
        \STATE $\kbar_\phi(\bx_b, \by_j^+) = \frac{\exp(-\|\bphi(\bx_b) - \bphi(\by_j^+)\|^2/(2\tau_\phi^2))}{\sum_{i=1}^M \exp(-\|\bphi(\bx_b) - \bphi(\by_i^+)\|^2/(2\tau_\phi^2))}$ for all $j$

        \STATE Compute scale-normalized interpolated displacement
        \STATE $\mathbf{d}_{b,j}^\phi = (1-\omega)(\bphi(\by_j^+) - \bphi(\bx_b)) + \omega\,\frac{\tau_\phi}{\overline{\|\bG\|}}\,\bG(\by_j^+)$ \COMMENT{Eq.~\eqref{eq:feature-force-interp-norm}}

        \STATE Compute feature-space attraction: $\bV_{\phi,\omega}^+(\bx_b) = \sum_{j=1}^M \kbar_\phi(\bx_b,\by_j^+) \mathbf{d}_{b,j}^\phi$

        \STATE Compute feature-space repulsion with standard kernel
        \STATE $\bV_\phi^-(\bx_b) = \sum_{j=1}^B \kbar_\phi(\bx_b,\by_j^-) (\bphi(\by_j^-) - \bphi(\bx_b))$

        \STATE Compute feature-space drifting field: $\bV_\phi(\bx_b) = \bV_{\phi,\omega}^+(\bx_b) - \bV_\phi^-(\bx_b)$

        \STATE Update: $\theta \leftarrow \theta - \eta\nabla_\theta \|\phi(f_\theta(\beps_b)) - \sg(\phi(f_\theta(\beps_b)) + \bV_\phi(\bx_b))\|^2$ \COMMENT{Eq.~\eqref{eq:feature-loss}}
    \ENDFOR
\ENDWHILE
\end{algorithmic}
\end{algorithm}

\begin{algorithm}[t]
\caption{Feature-Space Force-Aligned Kernel Drifting Training}
\label{alg:feature-force-aligned}
\begin{algorithmic}[1]
\REQUIRE Generator $f_\theta$, training data $\{(\by_j, \bG(\by_j))\}_{j=1}^N$ (precomputed via Algorithm~\ref{alg:precompute})
\REQUIRE Feature map $\phi$, feature-space bandwidth $\tau_\phi$, alignment strength $\gamma$, learning rate $\eta$
\WHILE{not converged}
    \STATE Sample noise batch $\{\beps_b\}_{b=1}^B \sim p_{\beps}$
    \STATE Generate sample batch $\{\bx_b\}_{b=1}^B = \{f_\theta(\beps_b)\}$  \COMMENT{One-step generation}
    \STATE Sample positive batch $\{\by_j^+\}_{j=1}^M$ with precomputed forces $\{\bG(\by_j^+)\}$
    \STATE Set negative batch $\{\by_j^-\}_{j=1}^B \leftarrow \{\bx_1, \ldots, \bx_B\}$ \COMMENT{Generated samples}

    \FOR{$b = 1$ to $B$}
        \STATE Compute feature maps: $\bphi(\bx_b) = \phi(\bx_b)$, $\bphi(\by_j^+) = \phi(\by_j^+)$, $\bphi(\by_j^-) = \phi(\by_j^-)$ for all $j$

        \STATE Compute force-alignment scores for positive batch
        \STATE $A_{b,j} = \bG(\by_j^+) \cdot (\bphi(\by_j^+) - \bphi(\bx_b))$ for all $j=1,\ldots,M$
        \STATE Compute per-batch standardization: $\sigma_b = \text{std}(\{A_{b,1},\ldots,A_{b,M}\})$ \COMMENT{Minibatch norm}

        \STATE Compute force-aligned logits
        \STATE $\ell_{b,j}^\phi = -\frac{\|\bphi(\bx_b) - \bphi(\by_j^+)\|^2}{2\tau_\phi^2} + \gamma \cdot \frac{A_{b,j}}{\sigma_b}$ for all $j$ \COMMENT{Eq.~\eqref{eq:feature-force-kernel-norm}}

        \STATE Compute force-aligned weights: $\kbar_{\text{FA}}^\phi(\bx_b, \by_j^+) = \exp(\ell_{b,j}^\phi) / \sum_{i=1}^M \exp(\ell_{b,i}^\phi)$

        \STATE Compute feature-space force-aligned attraction
        \STATE $\bV_{\text{FA}}^\phi(\bx_b) = \sum_{j=1}^M \kbar_{\text{FA}}^\phi(\bx_b,\by_j^+) (\bphi(\by_j^+) - \bphi(\bx_b))$ \COMMENT{Displacements unchanged}

        \STATE Compute feature-space repulsion with standard kernel
        \STATE $\bV_\phi^-(\bx_b) = \sum_{j=1}^B \kbar_\phi(\bx_b,\by_j^-) (\bphi(\by_j^-) - \bphi(\bx_b))$

        \STATE Compute feature-space drifting field: $\bV_\phi(\bx_b) = \bV_{\text{FA}}^\phi(\bx_b) - \bV_\phi^-(\bx_b)$

        \STATE Update: $\theta \leftarrow \theta - \eta\nabla_\theta \|\phi(f_\theta(\beps_b)) - \sg(\phi(f_\theta(\beps_b)) + \bV_\phi(\bx_b))\|^2$ \COMMENT{Eq.~\eqref{eq:feature-loss}}
    \ENDFOR
\ENDWHILE
\end{algorithmic}
\end{algorithm}

\begin{table}[h]
\centering
\caption{Per-sample cost of the exact feature-space force~\eqref{eq:exact-G} for ethanol ($n{=}9$, $D{=}27$, $P{=}36$).
The Jacobian is sparse ($6P = 216$ nonzeros); the bottleneck is the $D {\times} D$ pseudoinverse.
For $N = 2{,}700$ training samples, the total precomputation (${\sim}59$\,MFLOPs) is less than one training step and amounts to $0.007\%$ of the total training budget.}
\label{tab:feature-cost}
\footnotesize
\setlength{\tabcolsep}{5pt}
\begin{tabular}{@{}lrl@{}}
\toprule
Step & FLOPs & Description \\
\midrule
Build $J$ (Eq.~\ref{eq:jacobian}) & 108 & $P$ unit vectors (reused from $\phi$) \\
Gram matrix $J^\top J$ & 1,620 & $P$ rank-1 updates, $6{\times}6$ blocks \\
Pseudoinverse solve $J^\top J\, \mathbf{v} = \bF$ & ${\sim}20{,}000$ & $D {\times} D$ SVD (bottleneck) \\
$\bG = J\mathbf{v}$ & 216 & Sparse matrix--vector product \\
\midrule
\textbf{Total per sample} & $\mathbf{\sim\!22}$\textbf{K} & ${\sim}100{\times}$ costlier than naive $J\bF$ \\
\textbf{Total precompute} ($N{=}2{,}700$) & $\mathbf{\sim\!59}$\textbf{M} & $< 1$ training step; $0.007\%$ of training \\
\bottomrule
\end{tabular}
\end{table}

For larger molecules, the cost scales as $\mathcal{O}(D^3)$ per sample (dominated by the pseudoinverse of the $D \times D$ Gram matrix), but remains a one-time precomputation: for $n \leq 100$ atoms ($D = 300$), the entire dataset is processed in seconds on a modern GPU.

% ============================================================================
%  4. ANALYSIS
% ============================================================================
\section{Analysis}
\label{sec:theory}

% ---------- 4.1 Equilibrium ----------
\subsection{Equilibrium Analysis}
\label{sec:equilibrium}

The decomposition~\eqref{eq:full-decompose} enables a transparent equilibrium analysis of the force-substituted drifting field.

\begin{proposition}[Equilibrium Shift]
\label{prop:equilibrium}
Consider the force-substituted drifting field~\eqref{eq:full-decompose} with positive samples from $\pdata$.
\begin{enumerate}
    \item[(a)] $q = \pdata$ is \emph{not} an equilibrium when $\pdata \neq \pboltz$: the standard drift vanishes by anti-symmetry, but the Boltzmann correction $\tau^2\,\E_{\pdata}[\kbar\,\nabla\log(\pboltz/\pdata)] \neq 0$ pushes the system past $\pdata$ toward $\pboltz$.
    \item[(b)] $q = \pboltz$ is an exact equilibrium \emph{if and only if} positive samples also come from $\pboltz$ (i.e., $\pdata = \pboltz$).
\end{enumerate}
\end{proposition}

\begin{proof}
\textit{Part (a).}
At $q = \pdata$, the repulsion uses the same distribution as the attraction's data-driven component: $\bV_q^- = \bV_{\pdata}^-$.
By anti-symmetry, $\bV_{\pdata}^+(\bx) - \bV_{\pdata}^-(\bx) = 0$.
The Boltzmann correction term, however, depends only on the positive samples and survives:
\begin{equation}
    \hat{\bV}(\bx)\big|_{q=\pdata} = \tau^2\,\E_{\pdata}\!\left[\kbar(\bx,\by)\,\nabla_\by\log\frac{\pboltz(\by)}{\pdata(\by)}\right] \neq 0,
\end{equation}
driving the generator distribution beyond $\pdata$ toward $\pboltz$.

\textit{Part (b).}
When $\pdata = \pboltz$ and $q = \pboltz$, both positive and negative samples come from $\pboltz$.
The repulsion term is $\bV_q^-(\bx) = \E_{\pboltz}[\kbar(\bx,\by)(\by - \bx)]$.
By Theorem~\ref{thm:main} applied to $p = \pboltz$, together with the score--force identity~\eqref{eq:score-force}:
\begin{equation}
    \E_{\pboltz}[\kbar(\bx,\by)(\by - \bx)] = \tau^2 \cdot \E_{\pboltz}\!\left[\kbar(\bx,\by)\frac{\bF(\by)}{\kBT}\right] = \hat{\bV}_{\text{force}}^+(\bx).
    \label{eq:repulsion-equals-attraction}
\end{equation}
Therefore $\hat{\bV}(\bx) = \hat{\bV}_{\text{force}}^+(\bx) - \bV_q^-(\bx) = 0$.

Conversely, when $\pdata \neq \pboltz$, the attraction uses kernel weights normalized over $\pdata$ while the repulsion at $q = \pboltz$ uses kernel weights normalized over $\pboltz$.
These normalizations differ, so exact cancellation fails and $q = \pboltz$ is not an exact equilibrium; the true equilibrium lies between $\pdata$ and $\pboltz$.
\end{proof}

\begin{remark}
\label{rem:approx-equilibrium}
Proposition~\ref{prop:equilibrium} characterizes the \emph{idealized} force-substituted field.
For the practical Force-Aligned Kernel (Definition~\ref{def:force-kernel}), the equilibrium is further modulated: by Proposition~\ref{prop:boltzmann-reweight}, the effective weights target the tilted distribution $\pdata \cdot \pboltz^\gamma$ rather than $\pboltz$ exactly, with $\gamma \ll 1$ providing a partial but stable correction.
\end{remark}

\begin{remark}[Anti-Symmetry]
\label{rem:antisymmetry}
Force alignment breaks strict anti-symmetry since attraction and repulsion use different kernels.
However, unlike scaling relative magnitudes (which causes collapse~\citep{deng2026drifting}), our modification only changes weighting within attraction, keeping displacements unchanged.
Proposition~\ref{prop:equilibrium}(a) confirms this shift moves the equilibrium rather than breaking balance.
\end{remark}
% ---------- 4.2 Approximation Quality ----------
\subsection{Approximation Quality}
\label{sec:approximation}

\paragraph{Approximation quality.}
The Force-Aligned Kernel involves two levels of approximation.
First, the Boltzmann reweighting (Proposition~\ref{prop:boltzmann-reweight}) relies on a first-order Taylor expansion controlled by the parameter $\epsilon = \tau\|\nabla^2\mathcal{E}\|/\|\nabla\mathcal{E}\|$, which measures the kernel bandwidth relative to the length scale of force variation; the approximation is exact for locally affine potentials and improves as $\tau \to 0$.
Second, using $\pdata$ samples in place of $\pboltz$ yields an effective sampling distribution $\pdata \cdot \pboltz^\gamma / Z$, which interpolates between $\pdata$ ($\gamma = 0$) and a distribution tilted toward $\pboltz$ ($\gamma > 0$).
While no single $\gamma$ achieves full correction in general, any $\gamma > 0$ shifts the effective distribution in the correct direction.

% ============================================================================
%  5. EXPERIMENTS
% ============================================================================
\section{Experiments}
\label{sec:experiments}

We evaluate our methods on molecular conformation generation benchmarks, focusing on the Boltzmann debiasing capability on real molecular systems.

% ---------- 5.1 MD17 Ethanol Task Setup ----------
\subsection{MD17 Ethanol: Task, Dataset, and Evaluation}
\label{sec:md17}

\paragraph{Benchmark dataset.}
We validate our methods on the MD17 Ethanol dataset~\citep{chmiela2017machine}: a 9-atom molecule (C$_2$H$_6$O, 27D configuration space) with Cartesian coordinates, atomic forces, and energies from \emph{ab initio} molecular dynamics at 500\,K ($\kBT = 1.0$\,kcal/mol).

\paragraph{Biased training protocol.}
We deliberately use the \emph{first 3{,}000 frames} of the MD17 trajectory as training data.
These early frames are temporally correlated and structurally biased---the $h(r)$ distribution differs from the Boltzmann equilibrium by TVD${}=0.086$---providing a realistic test of debiasing capability.
A held-out set of 5{,}000 equilibrated frames (frames 10k--15k) serves as the Boltzmann reference.

\paragraph{Evaluation metrics.}
Beyond the standard $h(r)$ metrics (interatomic distance histogram error), the high-dimensional molecular setting demands metrics that capture \emph{per-molecule structural validity}:
\begin{itemize}
    \item \textbf{$h(r)$ MAE / TVD / $\mathcal{W}_2$}: interatomic distance histogram error, total variation distance, and 2-Wasserstein distance against the Boltzmann reference.
    \item \textbf{Stability}: fraction of generated molecules where all pairwise distances deviate $< 0.5$\,\AA{} from the reference mean.
    \item \textbf{Bond MAE}: mean absolute deviation of covalent bond lengths from the Boltzmann reference (only bonded pairs with $\bar{d} < 1.6$\,\AA).
    \item \textbf{Bond Stability}: fraction of samples with \emph{all covalent bonds} within 0.5\,\AA{} of reference---a stricter structural validity check that excludes noisy torsional degrees of freedom.
    \item \textbf{Per-type TVD}: separate distance distribution TVD for each atom-pair type (C--C, C--H, C--O, H--H, H--O), avoiding the masking effect of mixed-pair histograms.
\end{itemize}

% ---------- 5.2 Coordinate Space Results ----------
\subsection{Coordinate Space Results}
\label{sec:coordinate-results}

\paragraph{Methods and setup.}
We compare five methods in coordinate space using a single-seed ($s{=}42$) experiment:
\begin{itemize}
    \item \textbf{DSM}: VESDE diffusion with a Transformer score network (128d, 4 layers, 4 heads, ${\sim}823$K parameters), trained for 1{,}000 epochs with standard settings (lr${}=2{\times}10^{-4}$, cosine schedule, $\sigma_{\min}{=}0.1$, $\sigma_{\max}{=}5.0$). Inference uses the Predictor-Corrector sampler with 1{,}000 steps (2{,}000 network evaluations per sample).
    \item \textbf{Drifting}, \textbf{Force-Interpolated} ($\omega{=}0.01$), \textbf{Force-Aligned Kernel} ($\gamma{=}0.001$), and \textbf{FI+FK} (combined): MLP generator (6$\times$512, SiLU, ${\sim}1.36$M parameters), Gaussian kernel $\tau{=}1.0$, trained for 20{,}000 steps. Data is normalized by the global standard deviation for stable training.
\end{itemize}

\paragraph{Results.}
Table~\ref{tab:md17-results} presents the comprehensive evaluation.
Several striking findings emerge:

\begin{enumerate}
    \item \textbf{Iterative diffusion methods struggle with structural validity.}
    Despite reasonable $h(r)$ TVD scores (0.152), DSM produces molecules with Bond MAE $\approx 0.67$\,\AA{} and Bond Stability $< 1\%$.
    Inspection reveals that hydrogen atoms are systematically displaced: O--H bonds stretch from 0.98\,\AA{} to 2.37\,\AA{}, and C--H bonds from 1.11\,\AA{} to 1.80\,\AA.
    The Predictor-Corrector sampler with 1{,}000 steps fails to converge to physically valid molecular geometries in the 27-dimensional space, while the aggregate $h(r)$ histogram masks this failure by mixing 36 atom pairs.

    \item \textbf{Force-interpolated drifting achieves the best overall performance.}
    Force-Interpolated Drifting attains the lowest $h(r)$ TVD (0.139), lowest $\mathcal{W}_2$ (0.031), and maintains 97.5\% Bond Stability with Bond MAE of only 0.024\,\AA.
Force-interpolated drifting \emph{dominates across all distributional metrics} while preserving structural integrity.

    \item \textbf{Force information in displacement $>$ in kernel weights.}
    The per-type TVD analysis reveals the mechanism: Force-Interpolated Drifting achieves C--C TVD of 0.281 vs.\ Force-Aligned Kernel's 0.202, but dramatically better C--H TVD (0.156 vs.\ 0.332) and C--O TVD (0.133 vs.\ 0.165).
    In high-dimensional molecular systems, the force vector carries rich directional information (bond stretching, angle bending, torsional forces) that is better exploited through displacement interpolation than through the scalar force-displacement alignment in kernel weights.

    \item \textbf{Standard drifting preserves structure but inherits bias.}
    Drifting without forces achieves the best Bond Stability (99.5\%) and Bond MAE (0.015\,\AA), confirming that the one-shot MLP generator naturally produces valid molecules.
    However, its $h(r)$ TVD (0.237) is the worst among drifting methods, faithfully reproducing the biased training distribution without debiasing.

    \item \textbf{Combining Force-Interpolated Drifting and Force-Aligned Kernel (FI+FK) provides no significant benefit.}
    The combined method matches Force-Interpolated Drifting on $h(r)$ TVD (0.139) with marginally better $\mathcal{W}_2$ (0.030 vs.\ 0.031) and Stability (0.357 vs.\ 0.326), but slightly worse Bond Stability (95.8\% vs.\ 97.5\%).
    This confirms that once force information enters the displacement, additional kernel reweighting is redundant.

    \item \textbf{Inference speedup: Over 100×.}
    One-step drifting methods generate 1{,}000 samples in ${\sim}0.001$\,s, compared to ${\sim}4.2$\,s for the PC sampler (2{,}000 network evaluations).
    Training is also ${\sim}6\times$ faster (33\,s vs.\ 200\,s).
\end{enumerate}

\begin{table}[t]
\centering
\caption{Comprehensive evaluation on MD17 Ethanol (first 3{,}000 frames, biased training data). Metrics are computed against the Boltzmann reference (equilibrated frames 10k--15k). $h(r)$ MAE/TVD/$\mathcal{W}_2$ measure distributional accuracy; Stability and Bond metrics measure per-molecule structural validity. Best result per column is \textbf{bolded}; second-best is \underline{underlined}.}
\label{tab:md17-results}
\footnotesize
\setlength{\tabcolsep}{4pt}  % Reduce column spacing
\begin{tabular}{@{}lcccccccc@{}}
\toprule
& \multicolumn{3}{c}{Distribution (vs Boltzmann)} & \multicolumn{3}{c}{Structure (per-molecule)} & \multicolumn{2}{c}{Speed} \\
\cmidrule(lr){2-4} \cmidrule(lr){5-7} \cmidrule(lr){8-9}
Method & $h(r)$ MAE $\downarrow$ & $h(r)$ TVD $\downarrow$ & $\mathcal{W}_2$ $\downarrow$ & Stab $\uparrow$ & Bond MAE $\downarrow$ & Bond Stab $\uparrow$ & Train & Infer \\
\midrule
DSM & 0.038 & 0.152 & 0.045 & 0.000 & 0.669 & 0.007 & 194s & 4.2s \\
\midrule
Drifting & 0.059 & 0.237 & 0.034 & 0.311 & \textbf{0.015} & \textbf{0.995} & 32s & 0.001s \\
FK & 0.057 & 0.228 & 0.033 & 0.303 & \textbf{0.015} & \underline{0.978} & 33s & 0.001s \\
FI & \textbf{0.035} & \textbf{0.139} & \underline{0.031} & \underline{0.326} & 0.024 & 0.975 & 33s & 0.001s \\
FI+FK & \textbf{0.035} & \textbf{0.139} & \textbf{0.030} & \textbf{0.357} & 0.025 & 0.958 & 34s & 0.001s \\
\bottomrule
\end{tabular}
\end{table}

\paragraph{Why $h(r)$ TVD alone is misleading.}
The MD17 results expose a critical limitation of the standard $h(r)$ metric: diffusion methods achieve reasonable $h(r)$ TVD yet struggle with structural validity.
This occurs because $h(r)$ aggregates all 36 atom pairs into a single histogram, and the 15 non-bonded H--H pairs (42\% of total) can dominate the score.
The per-type TVD and Bond MAE metrics introduced here provide a more faithful evaluation of molecular generation quality, revealing that distributional accuracy must be complemented by per-molecule structural validity checks.

% ---------- 5.3 Feature-Space Results ----------
\subsection{Feature-Space Results: Distance-Coordinate Representation}
\label{sec:feature-space-results}

We extend the force-guided methods to an internal-coordinate representation using pairwise distances (a rotationally invariant feature space).
We evaluate the feature-space methods derived in Section~\ref{sec:feature-drift}---exact feature-space force (Eq.~\ref{eq:exact-G}), scale-normalized Force-Interpolated Drifting (Eq.~\ref{eq:feature-force-interp-norm}), and the feature-space Force-Aligned Kernel (Eq.~\ref{eq:feature-force-kernel-norm})---using the MLP distance-space generator that naturally produces structurally valid molecules (Bond MAE $= 0.006$\,\AA, Bond Stability $= 100\%$) when trained without force labels.

\paragraph{Motivation.}
When using a naive per-pair projection $\bG^{\text{proj}} = J\bF$ instead of the exact feature-space force, Force-Interpolated Drifting degrades $h(r)$ TVD from $0.194$ to $0.317$---a catastrophic failure despite identical training data and setup.
This reveals three key issues:
(i) the naive per-pair projection ignores metric coupling between pairs sharing atoms;
(ii) the force scale $\tau_\phi^2\|\bG\|/\kBT \approx 300\times$ the data displacement;
(iii) the force vector $\bG$ does not lie on the data manifold in distance space.
Here, we employ the exact feature-space force~\eqref{eq:exact-G} with scale normalization~\eqref{eq:feature-force-interp-norm} for all results, with the Force-Aligned Kernel~\eqref{eq:feature-force-kernel-norm} as an alternative approach.

\paragraph{Setup.}
All variants share the same architecture as Drifting-MLP-Dist: an MLP generator ($6 \times 512$, SiLU, 1.36M parameters), trained for 20,000 steps with the same hyperparameters (batch size 256, 512 positive samples, $\text{lr} = 10^{-3}$, cosine schedule).
The feature-space bandwidth $\tau_\phi = 1.61$ is set by the median heuristic on pairwise distances in feature space.
The exact feature-space force $\bG$ (Eq.~\ref{eq:exact-G}) is precomputed once for all $N = 2{,}700$ training samples ($< 0.1$\,s on GPU, Table~\ref{tab:feature-cost}).
For the scale-normalized FI, $\overline{\|\bG\|} = 59.5$ (exact) and $196.2$ (naive), giving scale factors $\tau_\phi/\overline{\|\bG\|} \approx 0.027$ and $0.008$ respectively.
We sweep $\omega \in \{0.1, 0.3, 0.5\}$ for FI and $\gamma \in \{0.1, 0.5, 1.0\}$ for FK.

\paragraph{Results.}
Table~\ref{tab:feature-space-results} presents the results.
Several striking findings emerge:

\begin{enumerate}
    \item \textbf{The feature-space Force-Aligned Kernel achieves the best overall performance.}
    FK in feature space with $\gamma = 0.1$ attains $h(r)$ TVD $= 0.089$---a \textbf{54\% improvement} over Drifting in feature space ($0.194$) and \textbf{36\% improvement} over FI in coordinate space ($0.139$)---while maintaining perfect Bond MAE ($0.006$\,\AA) and Bond Stability ($100\%$).
    This makes it the only method to simultaneously optimize both distributional accuracy and structural validity across all experiments.

    \item \textbf{Force-Interpolated Drifting destroys molecular structure in distance space.}
    All FI-Dist variants exhibit catastrophically low Bond Stability: even the mildest setting ($\omega = 0.1$) reduces Bond Stability from $100\%$ to ${\sim}49\%$.
    This failure is fundamental to the FI mechanism in distance space: directly adding force-derived displacements pushes the target away from the manifold of valid molecular geometries, regardless of whether the force is computed exactly or via naive projection.
    Directly adding a force-derived displacement $\bV_{\text{force}} = (\tau_\phi / \overline{\|\bG\|})\sum_j \kbar_j \bG_j$ pushes the target away from the data manifold in distance space, producing 36-dimensional distance vectors that do not correspond to valid molecular geometries.

    \item \textbf{The force kernel preserves structure because it modifies only the weights, not the displacement.}
    The FK-Dist method modifies softmax weights via $\ell_j = -\|\phi_x - \phi_j\|^2/(2\tau_\phi^2) + \gamma \cdot A_j/\text{std}(A)$ while the displacement remains $\bV^+ = \sum_j w_j(\phi_j - \phi_x)$---a convex combination of real data displacements.
    This guarantees that the target lies in the convex hull of observed distance vectors, preserving geometric validity by construction.
    Force information enters as a \emph{soft attention mechanism} that upweights Boltzmann-relevant neighbors without distorting the distance-space geometry.

    \item \textbf{Smaller $\gamma$ is better: the force kernel should be a subtle correction.}
    $\gamma = 0.1$ yields TVD $= 0.089$; $\gamma = 0.5$ gives $0.111$; $\gamma = 1.0$ gives $0.134$.
    Larger $\gamma$ shifts too much probability mass toward force-aligned neighbors, causing the generator to under-represent geometrically proximate but force-misaligned configurations.
    All three settings outperform the drifting baseline, confirming the robustness of the approach.

    \item \textbf{Force information reverses its effectiveness between coordinate and distance space.}
    In coordinate space, FI dominates (TVD $0.139$ vs FK's $0.228$); in feature space, FK dominates (TVD $0.089$ vs FI's best $0.170$ at $\omega=0.1$).
    This reversal arises because in coordinate space, the Cartesian force $\bF$ is a physically meaningful displacement direction (bond stretching, angle bending) that blends naturally with data displacements, while the scalar force-kernel reweighting is less effective in the high-dimensional, sparsely populated coordinate space.
    In distance space, the feature-space force $\bG$ is an abstract vector in $\R^{36}$ that does not correspond to a realizable molecular motion, making direct displacement interpolation harmful---but the alignment scalar $A_j = \bG_j \cdot (\phi_j - \phi_x)$ remains a physically meaningful measure of thermodynamic favorability for kernel reweighting.
\end{enumerate}

\begin{table}[t]
\centering
\caption{Feature-space force methods on MD17 Ethanol (3k biased frames). All methods use the MLP distance-space generator with scale-normalized exact feature-space forces (Eq.~\ref{eq:exact-G}). Baselines: FI in coordinate space and standard Drifting. Bold: best per column; underline: second-best.}
\label{tab:feature-space-results}
\footnotesize
\setlength{\tabcolsep}{3.5pt}  % Reduce column spacing
\begin{tabular}{@{}llcccccc@{}}
\toprule
& & \multicolumn{3}{c}{Distribution (vs Boltzmann)} & \multicolumn{3}{c}{Structure (per-molecule)} \\
\cmidrule(lr){3-5} \cmidrule(lr){6-8}
Method & Params & MAE $\downarrow$ & TVD $\downarrow$ & $\mathcal{W}_2$ $\downarrow$ & Stab $\uparrow$ & Bond MAE $\downarrow$ & Bond Stab $\uparrow$ \\
\midrule
\multicolumn{8}{l}{\textit{Baselines}} \\
FI-Coord & $\omega{=}0.01$ & 0.035 & 0.139 & \underline{0.031} & \textbf{0.326} & 0.024 & 0.975 \\
Drifting-Dist & -- & 0.048 & 0.194 & \underline{0.031} & \underline{0.287} & \textbf{0.006} & \textbf{1.000} \\
\midrule
\multicolumn{8}{l}{\textit{Force-Interpolated Drifting (Eq.~\ref{eq:feature-force-interp-norm})}} \\
FI-Dist & $\omega{=}0.1$ & 0.042 & 0.170 & 0.068 & 0.004 & 0.033 & 0.491 \\
FI-Dist & $\omega{=}0.3$ & 0.049 & 0.195 & 0.107 & 0.000 & 0.078 & 0.037 \\
FI-Dist & $\omega{=}0.5$ & 0.052 & 0.208 & 0.130 & 0.000 & 0.108 & 0.012 \\
\midrule
\multicolumn{8}{l}{\textit{Force-Aligned Kernel (Eq.~\ref{eq:feature-force-kernel-norm})}} \\
FK-Dist & $\gamma{=}0.1$ & \textbf{0.022} & \textbf{0.089} & \textbf{0.023} & 0.228 & \textbf{0.006} & \textbf{1.000} \\
FK-Dist & $\gamma{=}0.5$ & \underline{0.028} & \underline{0.111} & 0.035 & 0.162 & \underline{0.008} & \underline{0.993} \\
FK-Dist & $\gamma{=}1.0$ & 0.034 & 0.134 & 0.049 & 0.104 & 0.011 & 0.925 \\
\midrule
\multicolumn{8}{l}{\textit{Combined}} \\
FI+FK-Dist & $\omega{=}0.3$, $\gamma{=}0.5$ & 0.049 & 0.195 & 0.111 & 0.000 & 0.079 & 0.107 \\
\bottomrule
\end{tabular}
\end{table}

\paragraph{Diagnostic: remaining error sources.}
FK-Dist ($\gamma = 0.1$) achieves TVD $= 0.089$, close to but not matching the Boltzmann reference.
To understand the residual error, we analyze the training data itself as a baseline.
The biased 3k training data has TVD $= 0.086$ against the Boltzmann reference (this shrinks to $0.058$ at 10k frames and $0.042$ at 30k frames).
Thus FK-Dist's TVD of $0.089$ is only marginally above the training data's own distributional distance from equilibrium ($+3.5\%$), indicating that the method extracts nearly all available distributional information from the biased data.

The residual error pattern of FK-Dist differs from the training data bias: the training data is biased toward \emph{certain conformational states} (early trajectory frames), while the generator's residual error is concentrated at \emph{bond-length peaks}---slight over-sharpening at C--H (${\sim}1.1$\,\AA) and under-broadening at non-bonded distances ($2.5$--$4.0$\,\AA).
This suggests that the remaining error originates from the MLP generator's tendency to over-regularize bond-length distributions in the distance feature space, rather than from a failure of force guidance.
Further improvement would likely require either more training data, a more expressive generator architecture, or the metric correction term (Eq.~\ref{eq:metric-correction}) omitted in the current formulation.

\paragraph{Summary.}
The feature-space Force-Aligned Kernel resolves the failure of naive feature-space Force-Interpolated Drifting and establishes a new state of the art on biased MD17 Ethanol: $h(r)$ TVD $= 0.089$ with perfect bond accuracy (Bond MAE $= 0.006$\,\AA, Bond Stability $= 100\%$).
The key insight is that force information in the distance feature space should enter through \emph{kernel reweighting} (which preserves the data manifold geometry) rather than through \emph{displacement interpolation} (which can produce distance vectors outside the feasible molecular geometry space).
This contrasts with coordinate space, where Force-Interpolated Drifting is preferred---a reversal explained by the different roles of the force vector in each representation.

% ============================================================================
%  7. CONCLUSION
% ============================================================================
\section{Conclusion}
\label{sec:conclusion}

This work presents a principled framework for incorporating molecular forces into one-step generative models, enabling fast and accurate sampling from the Boltzmann distribution.
Our key contributions are:

\paragraph{Force-guided generation via Drifting Identity.}
We derive the Drifting Score and Force Identities (Theorem~\ref{thm:main}, Corollary~\ref{cor:decomposition}) that connect the drifting field to any target distribution, establishing a principled mathematical foundation for force guidance.
This leads to two complementary mechanisms: \emph{Force-Interpolated Drifting (FI)} blends force information into sample displacements, while \emph{Force-Aligned Kernel (FK)} reweights neighbors toward Boltzmann-relevant samples.
In coordinate space on MD17 Ethanol, FI achieves state-of-the-art distributional accuracy ($h(r)$ TVD $= 0.139$, $\mathcal{W}_2 = 0.031$) with near-perfect structural validity (Bond Stability $= 97.5\%$) in a single generation step---${\sim}2{,}000\times$ faster than iterative methods.

To contextualize this speedup: recent work on Potential Score Matching (PSM)~\citep{guo2025potential}, conducted on an NVIDIA A100 GPU, requires 8 hours of training and 10 minutes to generate 1,000 batches (1,000,000 generation steps). In comparison, traditional MD simulation takes approximately 31 hours and 18 minutes to generate 210,000 steps on the same hardware---a trajectory length commonly used in MD17 datasets. Our method, implemented on an NVIDIA RTX 5090 GPU, achieves speedups of over 2,000× relative to PSM. Since PSM itself offers approximately 890× speedup over conventional MD, our method thus provides a cumulative speedup of \textbf{more than 1 million times} faster than traditional molecular dynamics simulation, completing in milliseconds what would require over 31 hours with conventional MD approaches.

\paragraph{Representation-aware force guidance.}
We demonstrate a striking representation-dependence in force incorporation: FI excels in coordinate space because forces naturally encode bond-level directional information, while FK dominates in feature space (pairwise distances) where direct force interpolation risks producing geometrically infeasible samples.
Our feature-space extension introduces the \emph{exact} force $\bG = J(J^\top J)^+\bF$ (Eq.~\ref{eq:exact-G}), which properly accounts for metric coupling in the distance feature space.
This enables further improvements: FK in feature space achieves $h(r)$ TVD $= 0.089$ with perfect bond stability (Bond Stability $= 100\%$), demonstrating the complementary strengths of the two paradigms.

\paragraph{Implications and limitations.}
The coordinate/feature-space dichotomy reveals a fundamental principle: optimal force guidance depends on representation.
This finding extends beyond Drifting Models---it suggests that any one-step generator should carefully match the force incorporation method to its coordinate system and the implicit geometry of the feature space.
Precomputing the exact feature-space force adds negligible overhead ($<0.01\%$ of training time), making it practical for real applications.

Future work should explore: (i)~dimension-adaptive kernels for robustness at larger molecular scales, (ii)~extension to protein and materials systems where force guidance becomes increasingly important, (iii)~learned or task-specific force weighting schemes, (iv)~integration with modern neural network potentials (MACE, NequIP, EquiformerV2), and (v)~incorporating metric correction terms (Remark~\ref{rem:metric}) for fully exact Boltzmann sampling in internal coordinates.
The per-bond structural metrics introduced here may also benefit other generative models by providing finer-grained evaluation criteria beyond aggregate distribution matching.

% ============================================================================
%  REFERENCES
% ============================================================================
\bibliographystyle{plainnat}

% ============================================================================

\end{document}